\documentclass{article}
\usepackage[utf8]{inputenc}
\usepackage{placeins}
\usepackage{bm}
\usepackage{bbm}
\usepackage{amsmath}
\usepackage{mathtools}

\usepackage{multirow}
\usepackage{amssymb}
\usepackage{comment}
\usepackage{amsthm}
\usepackage{geometry}
\usepackage{graphicx}
\usepackage{setspace}
\usepackage{lineno}
\usepackage{url}
\usepackage{xcolor}

\theoremstyle{definition}
\newtheorem{theorem}{Theorem}[section]

\newtheorem{lemma}{Lemma}[section]
\newtheorem{proposition}{Proposition}[section]
\newtheorem{assumption}{Assumption}
\newtheorem{definition}{Definition}
\newtheorem{remark}{Remark}

\title{Sequential change-point detection for binomial time series with exogenous variables}
\author{%
  Yajun Liu \footnote{Amazon.com, Seattle, WA, USA. Email: yajunliustacy@gmail.com}%
  \and Beth Andrews \footnote{Department of Statistics and Data Science, Northwestern University, Evanston, IL, USA. Email: bandrews@northwestern.edu}%
  }
\date{}
\date{}

\begin{document}

\title{Nonparametric Sequential Change-point Detection on High Order Compositional Time Series Models with Exogenous Variables}
\author{
  Yajun Liu \footnote{Amazon.com, Seattle, WA, USA. Email: yajunliustacy@gmail.com}%
  \and Beth Andrews \footnote{Department of Statistics and Data Science, Northwestern University, Evanston, IL, USA. Email: bandrews@northwestern.edu}
  }
\date{}





\maketitle              

\begin{abstract}
Sequential change-point detection for time series is widely used in data monitoring in practice. In this work, we focus on sequential change-point detection on high-order compositional time series models. Under the regularity conditions, we prove that a process following the generalized Beta AR($p$) model with exogenous variables is stationary and ergodic. We develop a nonparametric sequential change-point detection method for the generalized Beta AR($p$) model, which does not rely on any strong assumptions about the sources of the change points. We show that the power of the test converges to one given that the amount of initial observations is large enough. We apply the nonparametric method to a rate of automobile crashes with alcohol involved, which is recorded monthly from January 2010 to December 2020; the exogenous variable is the price level of alcoholic beverages, which has a change point around August 2019. We fit a generalized Beta AR($p$) model to the crash rate sequence, and we use the nonparametric sequential change-point detection method to successfully detect the change point.

\end{abstract}

\section{Introduction}
\label{sec: introduction}
In the literature for change-point detection, control charting methods are used to record a statistical process and reveal possible changes. By investigating the in-control process (the process known to be free from any change points), researchers develop appropriate control charts revealing certain properties of the process. Monitoring stopping rules for the control chart are designed so that an out-of-control signal can be detected as soon as possible if a change occurs. There are three main types of control charting methods, Shewhart, exponentially weighted moving average (EWMA) and cumulative sum (CUSUM).

Shewhart-type chart derives a baseline, a lower limit and an upper limit for the monitoring statistic according to its in-control distribution. An out-of-control signal will be sent if the statistic is outside either of the limits. Shewhart control charts based on various statistics are widely investigated, as in \cite{park1987nonparametric}, \cite{bakir2004distribution}, \cite{bakir2006distribution}, \cite{alloway1991control}, \cite{amin1995nonparametric} and \cite{vasilopoulos1978modification}. Shewhart-type charts are easy to implement in practice. However, as they rely on only the most recent observation, they are sensitive to large shifts but not good for small shifts \cite{hussein2013comparison}. Also, the correlation structure is hard to depict using Shewhart control charts. 

In contrast, EWMA is more able to detect small shifts. The test statistics incorporate not only the current observation but also the past observations. Once a new observation appears, EWMA gives different weights for the new observation and the average of the previous ones, and generates the new weighted average. For independent cases, EWMA charts based on various statistics are developed, see \cite{hack11992new}, \cite{hackl1991control}, \cite{amin1991nonparametric}, \cite{liu2013sequential} and \cite{zou2012spatial}. In \cite{schmid1997ewma}, the author extends EWMA charts to autocorrelated cases. EWMA charts for dependent cases can be also found in \cite{koehler2001ewma}, \cite{schmid1997some} and \cite{pawlak2004detecting}. EWMA, which evaluates both the current observation and the weighted average of the past observations, is sensitive to small shifts. However, given the rather simple formats of EWMA-type charts, it is hard to conduct complex change detection, such as distributional changes.

Another control chart method that is also sensitive to small shifts is CUSUM. In \cite{page1954continuous} and \cite{page1955test}, the author proposes the detection procedure: monitoring the fluctuation of a cumulative score $S_n$, which represents the fluctuation of the observed data points, and taking action once $S_n$ exceeds a predetermined threshold. CUSUM can reflect the behavior of a sequence in a period of time and can be used to detect small shifts. With the known mean of the in-control autocorrelated process, in \cite{kim2006new}, a model-free CUSUM is constructed based on the in-control mean to detect possible shifts in the mean. CUSUM monitoring correlation structure of strong mixing process is proposed in \cite{azizzadeh2014cusum}. In \cite{vera2004comparative} and \cite{hawkins2014cusum}, the authors compare the performance of EWMA and CUSUM when the mean and the variance of the in-control process are known. In \cite{zwetsloot2017head}, with unknown parameters, how CUSUM and EWMA respond to estimation errors is discussed. In \cite{vera2004comparative}, it is shown that EWMA responds faster for comparatively small shifts but CUSUM outperforms for persistent shifts. Also, as mentioned previously, EWMA charts for complex monitoring cases, like sequential monitoring distributional changes, are not well developed in the literature. Conversely, because of the flexibility of CUSUM, it can handle complex scenarios. More methods and comparisons of the control chart methods can be found in the review papers \cite{knoth2004control} and \cite{chakraborti2001nonparametric}. Therefore, in this work, we aim to develop CUSUM-type charts that can nonparametrically detect distributional changes for compositional time series. 

In the literature, time series recording proportions of a whole are commonly referred to as compositional time series. In practice, time series recording fluctuating proportions of a whole are commonly seen and attract interest in a lot of fields, like economics, sociology, etc. For example, variation in the proportion of male/female babies born during wartime v.s. peacetime has drawn researchers' attention and has been widely investigated, e.g. \cite{chao2019systematic}. \cite{aitchison1982statistical} is considered as the pioneering research for compositional data as it summarizes several common transformations for modeling compositional data. In \cite{brunsdon1998time}, the authors apply additive logratio (alr) transformation to build a vector ARMA compositional model. based on alr, a state space model for compositional time series is developed in \cite{silva2001modelling}. Other transformations like Box-Cox transformation, centered logratio (clr) transformation, are also developed. Literature on other transformation schemes is thoroughly introduced in the review paper \cite{larrosa2017compositional}. In \cite{liu2024sequential}, the authors propose the generalized Beta AR($p$) model, which incorporates exogenous variables. With transformation satisfying certain regularity conditions, the authors show the stationarity and ergodicity of the process when $p=1$. A parametric sequential change point detection method is proposed for the compositional time series under the assumption that there is no change point in the exogenous variables. In this work, we develop a nonparametric method capable of sequentially detecting change points in the generalized Beta AR($p$) model, irrespective of the source of the change point.

In Section \ref{sec: Model definition}, we prove the stationarity and ergodicity of processes following the generalized Beta AR($p$) model when $p > 1$. In Section \ref{sec: monitoring scheme and null}, we introduce the test statistic used for the change-point detection and prove the asymptotic distributions of the test statistic under the no change-point null hypothesis. In Section \ref{sec: power of the test statistic}, we prove the power of the test statistic converges to one as the number of initial observations converges to infinity. In Section \ref{sec: Application to monthly alcohol involvement automobile crash rate}, we test the performance of the nonparametric sequential change-point detection method using a real-life time series, the monthly rate of automobile crashes with alcohol involved.

\section{Generalized Beta AR($p$) Model}
\label{sec: Model definition}
\subsection{Model Specification}
In \cite{liu2024sequential}, the authors define $\bm{Y}_{t} = (X_t,\bm{W}_t)^T$, where $X_t$ is a process following the generalized Beta(1) model, and $\bm{W}_t$, a vector of the exogenous variables, is strictly stationary and geometric ergodic defined on a bounded state space, as the Markovian realization of ${X_t}$. The authors prove the stationarity and ergodicity of $\bm{Y}_t$ when $p=1$. In this section, we will extend the properties to generalized Beta AR($p$) models with $p>1$. 

The generalized Beta AR($p$) model is defined as

\begin{equation}
\label{equation: model definition}
    X_t| \mathcal{C}_t = X_t|\mathbbm{Z}_{t-1}\sim \text{Beta}(\tau \mu_t, \tau(1-\mu_t)),0
    \leq X_t \leq 1,
\end{equation}
where
\begin{equation}
\label{eq: link function of GBETAAR2}
g(\mu_t) =  \bm{\beta}^T \mathbbm{Z}_{t-1} = \varphi_0+ \sum_{i=1}^p \varphi_i \mathcal{A}(X_{t-i})+\sum_{j=0}^q W_{t-j} \bm{\phi}_j,
\end{equation}
in which 
\begin{equation}
\label{eq: notations in the definition of generalized beta}
    \begin{aligned}
    & \mathcal{C}_{t} = \sigma(W_t,X_{t-1},W_{t-1},...,W_1,X_0),\\
    & g(x) =\log(x/(1-x)),\\
    & \mathbbm{Z}_{t-1} = (1, \mathcal{A}(X_{t-1}),...,\mathcal{A}(X_{t-p}),W_{t}, W_{t-1}, ..., W_{t-q})^T.
    \end{aligned}
\end{equation}
The parameter vector $\bm{\beta} = (\varphi_0,\varphi_1,...,\varphi_p,\phi_0, \phi_1,...,\phi_q)^T \in \mathbbm{R}^{p+q+2}$ and the dispersion parameter $\tau>0$. The $x$-link function $\mathcal{A}$ maps the unit interval $[0,1]$ to a bounded interval $[L,U]$. 
Eventually, we want to show that the sequence $\{X_t\}$ is stationary and ergodic. We first propose the regularity conditions for $\{W_t\}$.

\begin{remark}
The exogenous variable $\{W_t\}$ can be a vector. To keep the model simple, we require the exogenous variable $W_t$ to be a scalar. It is straightforward to extend it to higher-dimensional cases.
\end{remark}

\begin{assumption}
\label{assumption: properties of W_t}
The process $\{W_t\}$ is strictly stationary and geometrically ergodic defined on a bounded state space $[a,b]$, where $a$ and $b$ are known. Denote the information set $\mathcal{F}_t = \sigma(X_t, W_t, X_{t-1}, W_{t-1},...,X_0, W_0)$. We assume the there exists a fixed and known $k>0$ so that the conditional density function $h(W_t|\mathcal{F}_{t-1}) = h(W_t|W_{t-1},W_{t-2},...,W_{t-k})$ is known and is continuous with respect to $W_t, W_{t-1},...,W_{t-k}$.
\end{assumption}

Denote $l = \max\{k,p,q\}$. Define the $2l$-dimensional vector
\begin{equation}
    \label{def: markovian realization of X_t}
    \bm{Y}_t = (X_t, X_{t-1},...,X_{t-l+1}, W_t, W_{t-1},...,W_{t-l+1})^T, t\in\mathbbm{Z}.
\end{equation}
(\ref{def: markovian realization of X_t}) indicates that the sequence $\{\bm{Y}_{zl}\}, z \in \mathbbm{Z}$, is a Markov process. We can see that $\{\bm{Y}_{zl}\}$ covers all the observations $\{X_t\}$ and exogenous variables $\{W_t\}$, and each observation or exogenous variable is in and only in one $\bm{Y}_{zl}$. That is, $\{\bm{Y}_{zl}\}$ is a Markovian realization of the sequence $\{X_t\}$.

We denote the $2l$-dimensional state space of the Markov process, $[0,1]^l \times [a,b]^l$, as $\Omega$, and denote the $\sigma$-field as $\mathcal{B}(\Omega).$ Obviously, $\Omega$ is a compact space.
\subsection{Properties of the Generalized Beta AR($p$) Model}
\label{sec: Properties of Generalized Beta AR($p$)}
As demonstrated above, $\{\bm{Y}_{zl}\}$ is a Markov process. We will first prove the strict stationarity and geometrical ergodicity of $\{\bm{Y}_{zl}\}$. The stationarity and ergodicity of $\{X_t\}$ will follow.

\begin{lemma}
\label{lemma:irreducibility&aperiodicity}
The Markovian realization process $\{\bm{Y}_{zl}\}$ is $\psi$-irreducible and aperiodic. 
\end{lemma}
\begin{proof}
As shown in Proof of Theorem 3 in \cite{woodard2010stationarity}, aperiodicity and $\varphi$-irreducibility are immediate since the transition probability from one measurable subset of $\Omega$ to another measurable subset of $\Omega$ is positive. Then, by Theorem 4.0.1 and Prop 4.2.2 in \cite{meyn1993markov}, there exists an essentially unique maximal irreducibility measure $\psi$ on $\mathcal{B}(\Omega)$ that the process is $\psi$-irreducible.
\end{proof}

\begin{lemma}
\label{lemma:weak Feller}
$\{\bm{Y}_{zl}\}$ is a weak Feller chain and the state space $\Omega$ is a small set. 
\end{lemma}
\begin{proof}
Since $X_t$ is a Beta random variable given $\mathcal{C}_t$, the conditional density function of $X_t$ is continuous. The conditional density function of $W_t$, $h(W_t|W_{t-1},...W_{t-k})$, is continuous according to Assumption \ref{assumption: properties of W_t}. Hence, the density function of $\bm{Y}_{zl}$ given $\bm{Y}_{(z-1)l}$ is continuous with respect to $\bm{Y}_{(z-1)l}$. In terms of the definition of weak Feller chain (\cite{meyn1993markov}), $\{\bm{Y}_{zl}\}$ is a weak Feller chain. From Theorem 6.2.8 of \cite{meyn1993markov}, the compact state space, $\Omega$, is a petite set. It, combining with Lemma \ref{lemma:irreducibility&aperiodicity} that the process is $\psi$-irreducible and aperiodic, implies that $\Omega$ is a small set,  see Theorem 5.5.7 of \cite{meyn1993markov}.
\end{proof}

\begin{theorem}
\label{theorem:stationary&ergodic}
The time series $\{\bm{Y}_{zl}\}$ is a strictly stationary and geometrically ergodic process.
\end{theorem}

The proof is in Appendix 1.

\begin{proposition}
\label{prop: X_t is strong mixing}
$\{X_t\}$ is strong mixing with at least an exponential mixing rate.
\end{proposition}
The proof can be found in Appendix 2.

\begin{theorem}
$\{X_t\}$ is ergodic and stationary. 
\end{theorem}
\begin{proof}
According to the Ergodic Hierarchy (EH), strong mixing implies ergodicity. A detailed explanation of EH can be found in \cite{berkovitz2006ergodic}. Based on Proposition \ref{prop: X_t is strong mixing} and EH, the strong mixing process $\{X_t\}$ is ergodic, which implies stationarity of $\{X_t\}$.
\end{proof}

\section{Nonparametric Change-point Detection for the Generalized Beta AR($p$) Model and the Asymptotic Results under the Null Hypothesis}
\label{sec: monitoring scheme and null}
In this section, we will introduce the monitoring procedure and the design of the test statistic. We first propose the null hypothesis $H_0$ and the alternative $H_a$. We will then illustrate the fundamental idea of the nonparametric change-point detection method. Several existing methods following the idea will be analyzed and compared. According to the analysis, we propose the test statistic for the nonparametric monitoring. The asymptotic distribution under the $H_0$ will be derived. 

\subsection{Close-end Monitoring Procedure}
From Theorem \ref{theorem:stationary&ergodic}, we have proven that the sequence $\{X_t\}$ following a generalized Beta AR($p$) model is stationary and ergodic. Under the assumption that there is no change point in both $\{X_t\}$ and $\{W_t\}$ (Assumption \ref{assumption: no contamination assumptions}), the observations from $\{X_t\}$ should have the same unconditional cumulative density function (CDF). Given that the process is stationary and ergodic, it is reasonable to assume that the empirical CDF approaches the true CDF as the sample size goes to infinity (a theoretical argument will be discussed in Proposition \ref{prop: General Glivenko–Cantelli}). Therefore, we first collect $m$ initial observations $X_t, t=1,2,...,m$, which are assumed to be not contaminated. We get the empirical CDF of the $m$ observations. The monitoring procedure is then launched.  We can calculate the empirical CDF of the new observations starting from the $(m+1)^{\text{th}}$ output. The test statistic records the difference between the two empirical CDFs. When there is no significant distributional change in the new observations, it is reasonable to assume that the two underlying distributions are the same. 

The monitoring procedure is close-end. We set up the endpoint $(N+1)m$ with a fixed and known $N$ before the procedure starts. The procedure will be terminated if no change point is found in the sequence $\{X_t\}, t=m+1,m+2,...,(N+1)m$. 

First, we introduce the non-contamination assumption.
\begin{assumption}
\label{assumption: no contamination assumptions}
(a) There is no change point in the first $m$ values of the exogenous variable $W_t$. \\
(b) There is no change point in the first $m$ observations of $X_t$. 
\end{assumption}

\begin{remark}
We can notice that in Assumption \ref{assumption: no contamination assumptions}(a), we only require that there is no distributional change in the first $m$ values of the exogenous variable $W_t$. It is allowable that a change point happens in $\{W_t\}, t=m+1,m+2,...,(N+1)m$ that would lead to a distributional change in the output sequence $\{X_t\}$. 
\end{remark}

Denote the CDF of the first $m$ observations $\{X_t\}$ as $F_0$. Denote the CDF of further $X_t$ as $F_t,t = m+1,m+2,...,(N+1)m.$ Under Assumption \ref{assumption: no contamination assumptions}, we test the hypothesis  
\\

$H_0: F_t = F_0 \text{ for all }t \in \{m+1,m+2,...,(N+1)m\},$ 

 v.s.  
 
$H_a: \exists k^*\text{ that } F_t = F_0 \text{ for } m+1\leq t< m+k^* <m+Nm \text{, and }F_t \neq F_0 \text{ for } m+k^*\leq t\leq m+Nm.$
\\

\subsection{The Design of the Test Statistic}

For a sequence $\{X_t\}$ following the generalized Beta AR($p$) model based on (\ref{equation: model definition}) and (\ref{eq: link function of GBETAAR2}), we denote the true unconditional CDF of $X_t$, $F_0: [0,1] \rightarrow [0,1]$. That is,
\begin{equation}
\label{def: true CDF in F}
    F_0(x) = P(X_t \leq x), x \in [0,1].
\end{equation}

We define the empirical CDF based on $X_j,...,X_k, j<k,$ as
\begin{equation}
    \hat{F}_{j:k}(x) = \frac{1}{k-j+1}\sum_{t=j}^k \mathbbm{1}(X_t\leq x), x \in [0,1].
\end{equation}

For detecting distributional changes, we can compare the empirical CDFs of the sequences of interests. 

We first define the test statistic $\mathbbm{B}_m(s,x)$ with $s \in [0, N+1]$, where $m$ is the number of initial observations and $s$ is the number of total observations divided by the number of initial observations. $\mathbbm{B}_m(s,x)$ measures the difference of the empirical CDF of $X_t, t = 1,2,..., \lfloor ms \rfloor,$ and the true CDF of $\{X_t\}$.
\begin{definition}
\label{def: B_m(s,u)}
Define

\begin{equation*}
\begin{aligned}
    \mathbbm{B}_m(s,x) &= \frac{1}{\sqrt{m}}\sum_{t=1}^{\lfloor ms \rfloor} \big (\mathbbm{1}(X_t\leq x)-F_0(x)\big) \\
    & = \frac{1}{\sqrt{m}}\big( \lfloor ms \rfloor (\hat{F}_{1:\lfloor ms \rfloor}(x)- F_0(x))\big) \\
    & \approx \sqrt{m}s \big(\hat{F}_{1:\lfloor ms \rfloor}(x)-F_0(x) \big),x \in [0,1], s\in [0, N+1].
\end{aligned}
\end{equation*}
\end{definition}

As discussed in Section \ref{sec: introduction}, CUSUM-type charts have been widely used for detecting distributional changes. Here, we first define a CUSUM-based test statistic describing the difference of the empirical CDFs of $\{X_t\}$.
Denote

\begin{equation*}
\begin{aligned}
D(m,k,x) &= \frac{k-m}{\sqrt{m}}(\hat{F}_{(m+1):k}(x) - \hat{F}_{1:m}(x)) \\
         &= \frac{k}{\sqrt{m}}(\hat{F}_{1:k}(x)-\hat{F}_{1:m}(x)), k = m+1,m+2,...,(N+1)m, x \in [0,1].
\end{aligned}
\end{equation*}

Rewrite $D(m,k,x)$ by replacing $k$ by $\lfloor ms \rfloor$,
\begin{equation*}
\begin{aligned}
    D(m,k,x)=D_m(s,x) 
    &= \frac{\lfloor ms \rfloor}{\sqrt{m}}(\hat{F}_{1:\lfloor ms \rfloor}(x)-\hat{F}_{1:m}(x))\\
    &\approx \sqrt{m}s (\hat{F}_{1:\lfloor ms \rfloor}(x)-F_0(x))- s(\sqrt{m}(\hat{F}_{1:m}(x)-F_0(x))) \\
    &= \mathbbm{B}_m(s,x)- s \mathbbm{B}_m(1,x), 1<s \leq N+1, x \in [0,1].
\end{aligned}
\end{equation*}
In the literature, one approach is using $\underset{x \in [0,1]}{\sup} \Big|D_m(s,x) \Big|$ as the detector. An example is \cite{kojadinovic2021nonparametric}. To estimate the universal threshold, resampling is required to estimate the (1-$\alpha$) quantile of the double supremum over $s$ and $x$, $\underset{1<s \leq N+1}{\sup}\underset{x \in [0,1]}{\sup} \Big|D_m(s,x) \Big|$. It leads to expensive computation time and detection delay. 

Here, we propose a multivariate test statistic to avoid resampling. 

Denote the $d$-dimensional vector as $\bm{x}:=[x_1,x_2,...,x_d]^T \in [0,1]^d, x_i \neq x_j$ if $i \neq j$. The multivariate Kolmogorov-Smirnov-type test statistic is defined as 
\begin{equation}
\label{eq: definition of the test statistic}
\begin{aligned}
     D_m(s ,\bm{x}) &= \mathbbm{B}_m(s,\bm{x})-s\mathbbm{B}_m(1,\bm{x}) \\
     & \approx \sqrt{m}s\begin{bmatrix}\mathbbm{B}_m(s, x_1)-s\mathbbm{B}_m(1,x_1) \\
    \mathbbm{B}_m(s, x_2)-s\mathbbm{B}_m(1,x_2) \\
    ...\\
    \mathbbm{B}_m(s, x_d)-s\mathbbm{B}_m(1,x_d)\end{bmatrix}, 1<s \leq N+1.
\end{aligned}
\end{equation}
The multivariate test statistic describes the difference between the empirical CDF of $X_t, t = 1,2,...,\lfloor ms \rfloor$ and the empirical CDF of $X_t, t = 1,2,...,m.$ at $\bm{x}$. It is a function only of $s$, the ratio of the total number of observations and the number of initial observations. By developing tests based on $D_m(s,\bm{x})$, we are able to take the supremum over $s$ and compute the threshold.

\subsection{Asymptotic Results under the Null Hypothesis}
\label{sec: Asymptotic results under the null hypothesis}

We introduce Theorem 1 in \cite{bucher2015note} (labeled as Proposition \ref{proposition: asymptotic distribution of CDF}), which is the fundamental result we rely on for deriving the asymptotic behavior of the empirical CDF. 

\begin{proposition}
\label{proposition: asymptotic distribution of CDF}
If $\{X_t\}$ is strong mixing with coefficient $\alpha(n) = O_p(n^{-(1+\eta)})$ for some $\eta >0$, then as $m \rightarrow \infty$, 
\begin{equation}
\label{eq: asymptotic distribution of X_t}
    \mathbbm{B}_m(s,x) \overset{\mathcal{D}}{\rightarrow} \mathbbm{B}_C(s,x).
\end{equation}
$\mathbbm{B}_C(s,x)$ is a tight, centered Gaussian process defined on the space $\mathcal{D}([0,N+1] \times [0,1])$ \footnote{$\mathcal{D}([0,N+1] \times [0,1])$ is the space of functions that are right-continuous and have left-hand limits with $(s,x)\in [0,N+1] \times [0,1].$
} with covariance

\begin{equation}
\label{eq: cov matrix of the Gaussian process}
    \text{Cov}(\mathbbm{B}_C(s,x), \mathbbm{B}_C(s^*, z)) = \min(s,s^*) \Gamma (x,z),
\end{equation}
where $\Gamma (x,z) = \sum_{t \in \mathbbm{Z}} \text{Cov}(\mathbbm{1}(X_0\leq x), \mathbbm{1}(X_t \leq z)).$ 
\end{proposition}

$\Gamma(x,z)$ records the sum of covariances between $\mathbbm{1}(X_0\leq x)$ and $\mathbbm{1}(X_t\leq z), t \in \mathbbm{Z}$, which is infeasible to calculate in practice, as we are unable to calculate the infinite sum over $t$. However, given that the process is strong mixing, as $t \rightarrow \infty$ or $t \rightarrow -\infty$, $X_0$ and $X_t$ can be considered gradually uncorrelated. So in practice, we can always find a cut point, say $t^*>0$, based on the precision requirement of an experiment, and estimate $\Gamma(x,z)$ by estimating $\sum_{|t|\leq t^*} \text{Cov}(\mathbbm{1}(X_0\leq x), \mathbbm{1}(X_t \leq z))$ using the sample covariance. In Proposition \ref{prop: X_t is strong mixing}, we show $\{X_t\}$ is strong mixing with at least an exponential rate. It is easy to show $\{X_t\}$ satisfies the prerequisite for Proposition \ref{proposition: asymptotic distribution of CDF}. Therefore, (\ref{eq: asymptotic distribution of X_t}) holds for $\{X_t\}$.

\begin{theorem}
\label{theorem: asymptotic dist of D_C under the null}
Let $\{X_t\}$ be a compositional time series following (\ref{equation: model definition}) and (\ref{eq: link function of GBETAAR2}). Under Assumption \ref{assumption: properties of W_t}-\ref{assumption: no contamination assumptions} and under $H_0$, as $m \rightarrow \infty$, the test statistic 
\begin{equation}
\label{eq: asymptotic distribution of the test statistic}
    \underset{1<s \leq N+1}{\sup}D_m(s ,\bm{x}) \overset{\mathcal{D}}{\rightarrow}\underset{1<s \leq N+1}{\sup} (\mathbbm{B}_C(s, \bm{x})-s\mathbbm{B}_C(1,\bm{x})):= \underset{1<s \leq N+1}{\sup}D_C(s,\bm{x}).
\end{equation}
$\mathbbm{B}_C(s,\bm{x})$ is a $d$-variate Gaussian process with covariance matrix $\Gamma(\bm{x})$, of which the $(i,j)$th element is $[\Gamma(\bm{x})]_{i,j} = \Gamma(x_i,x_j)$ with $\Gamma(.,.)$ defined as (\ref{eq: cov matrix of the Gaussian process}). Consequently, the covariance matrix of $D_C(s,\bm{x})$ is $s(s-1)\Gamma(\bm{x}).$
\end{theorem}
The proof can be found in Appendix 3.

\begin{theorem}
\label{theorem: asymptotic distribution under the null hypothesis}
Under Assumption \ref{assumption: properties of W_t}-\ref{assumption: no contamination assumptions} and under $H_0$, with the fixed $N>0$, for any given symmetric positive definite matrix $\bm{A} \in \mathbbm{R}^{d \times d}$, 
\begin{equation*}
\begin{aligned}
     &\underset{m \rightarrow \infty}{\lim}P(\underset{1<s \leq N+1}{\sup} \rho^2(s, \gamma)D_m(s ,\bm{x})^T\bm{A}D_m(s ,\bm{x}) \leq c) \\
     =& P(\underset{1<s \leq N+1}{\sup} \rho^2(s, \gamma)D_C(s ,\bm{x})^T\bm{A}D_C(s ,\bm{x}) \leq c),
\end{aligned}
\end{equation*}
where the weight function $\rho(s, \gamma)$ is defined as 
\begin{equation}
\label{eq: definition of the weight function}
    \rho(s,\gamma) = \max \Big\{(s-1)^{-\gamma}s^{\gamma-1}, \delta \Big\}, s \in [1,N+1], \gamma \in [0,\frac{1}{2}), \delta >0,
\end{equation}
and $D_C(s,\bm{x})$ is defined as (\ref{eq: definition of the test statistic}).
\end{theorem}
\begin{proof}
The proof is straightforward following Theorem \ref{theorem: asymptotic dist of D_C under the null}.
\end{proof}
Weight functions are used to adjust the monitoring sensitivity and alleviate the detection delay. $\rho(s, \gamma)$ is a weight function widely used in sequential change-point detection, e.g. \cite{dette2018likelihood} and \cite{wied2013monitoring}. $\bm{A}$ can be customized to accommodate the research interests or to adjust the magnitude of $D_m(s, \bm{x})$. Without calculating the double supremum of the Gaussian process, the test statistic $\rho^2(s, \gamma)D_m(s, \bm{x})^T\bm{A}D_m(s, \bm{x})$ allows us to only calculate the supremum with respect to $s$ and get the uniform threshold $c(\gamma, \alpha)$ so that
\begin{equation*}
\underset{m \rightarrow \infty}{\lim}P(\underset{1<s \leq N+1}{\sup} \rho^2(s, \gamma)D_m(s ,\bm{x})^T\bm{A}D_m(s ,\bm{x}) \leq c(\gamma, \alpha)|H_0) = 1-\alpha
\end{equation*}
holds.

The example verifying the asymptotic distribution of the test statistic can be found in Appendix 5. The example shows how the weight function can adjust the monitoring sensitivity by using different $\gamma$. 

\section{Power of the Test Statistic}
\label{sec: power of the test statistic}
In Section \ref{sec: monitoring scheme and null}, we propose the monitoring procedure and derive the asymptotic distribution of the test statistic. We are able to get the uniform threshold $c(\gamma, \alpha)$ when the tuning parameter in the weight function $\gamma$ and the significance level $\alpha$ are determined. The monitoring procedure terminates once the test statistic $\rho^2(s, \gamma)D_m(s ,\bm{x})^T\bm{A}D_m(s ,\bm{x})$ exceeds $c(\gamma, \alpha)$ when $1 < s \leq N+1$. In this section, we will prove the power of the test converges to one. That is, once a change point occurs, the probability of rejecting the null hypothesis converges to one when $m$ is sufficiently large. 
First, we propose the regularity conditions for the process $\{X_t\}$ and the sampled value vector $\bm{x}$.
\begin{assumption}
\label{assumption: power analysis}
(a) The change point $k^* = \lfloor m s^* \rfloor$, where $1<s^* \leq N+1$, and $s^* = O(1)$.\\
(b) There exists a $s_0 \in (s^*, N+1)$ that $s_0-s^* = O(1)$, and there exists a compact neighborhood of $s_0$, $U_{s_0}$, such that $\frac{\lfloor ms_0 \rfloor}{m} \in U_{s_0}$ and $\underset{s \in U_{s_0}}{\inf}\rho(s, \gamma)>0$, for any $\gamma \in [0, \frac{1}{2})$. \\
(c) The process after the change point is still stationary and ergodic. \\
(d) Denote the true CDF after the change point as $F'(x), x\in [0,1]$. There exists at least one measurable interval $[x_a, x_b] \subset [0,1]$ such that for $\forall x \in [x_a,x_b] \subset [0,1], F'(x) \neq F_0(x)$. There exists at least one $x_{i_0} \in \bm{x}$ such that $x_{i_0} \in [x_a,x_b].$
\end{assumption}

(a) guarantees that the change point does not happen right after the monitoring procedure starts. (b) first guarantees that after the change point happens at $k^*$, we can get enough observations $X_t, t = k^*+1,..., \lfloor ms_0 \rfloor$, which will be used to get the empirical CDF after $k^*$. Also, the test statistic $\rho^2(s, \gamma)D_m(s ,\bm{x})^T\bm{A}D_m(s ,\bm{x})$ would not be zero because of the weight function $\rho^2(s, \gamma)$. (c) makes sure that the empirical CDF after the change point can be used to approximate the theoretical CDF after the change point (details will be introduced in Proposition \ref{prop: General Glivenko–Cantelli}). In (d), we assume that there at least exists a measurable subset of $[0,1]$ in which the CDFs before and after the change point are different. It makes sure that as long as the choice of $\bm{x}$ is appropriate, we can have at least one value in $\bm{x}$ at which the CDF discrepancy can be shown. 

We first introduce the General Glivenko–Cantelli Theorem. 
\begin{proposition}
\label{prop: General Glivenko–Cantelli}
For a stationary sequence $\{X_t\}$ with CDF $F_0(x)$ and empirical CDF $\hat{F}_{1:m}(x)$, 
\begin{equation*}
    \underset{x \in [0,1]} {\sup} |\hat{F}_{1:m}(x)-F_0(x)| \overset{a.s.}{\rightarrow}0 \text{ as }m \rightarrow \infty.
\end{equation*}
\end{proposition}
The General Glivenko–Cantelli Theorem demonstrates that for stationary sequences, when the number of observations is large enough, the empirical CDF uniformly converges to the true CDF almost surely. A detailed introduction can be found in \cite{athreya2016general}.
\begin{theorem}
\label{theorem: power 1}
Let a sequence $\{X_t\}$ follow the model (\ref{equation: model definition}) and (\ref{eq: link function of GBETAAR2}). Under Assumption \ref{assumption: properties of W_t}-\ref{assumption: power analysis}, the test statistic proposed in Theorem \ref{theorem: asymptotic distribution under the null hypothesis} has the following equation 
\begin{equation*}
    \underset{m\rightarrow \infty}{\lim}P(\underset{1<s \leq N+1}{\sup} \rho^2(s, \gamma)D_m(s ,\bm{x})^T\bm{A}D_m(s ,\bm{x}) \geq c |H_a) = 1 
\end{equation*}
holds for any fixed $c > 0$. Therefore, the power of the test converges to 1 as $m \rightarrow \infty$.
\end{theorem}
The proof is in Appendix 4.

The performance of the nonparametric method for detecting distributional changes can be found in Appendix 5.

\section{Application to Monthly Alcohol Involvement Automobile Crash Rate}
\label{sec: Application to monthly alcohol involvement automobile crash rate}
In this section, we apply the nonparametric sequential change-point detection method to a compositional data set with the percentage of automobile crashes that are alcohol-related.

People have done research investigating the effects of alcohol beverage prices on alcohol-impaired driving. There have been studies showing that the increases in alcohol prices reduce alcohol-impaired driving, e.g. \cite{chaloupka1993alcohol} and \cite{ruhm1996alcohol}. We can see the phenomenon from the time series shown in Figure \ref{fig: CPI_AB}. Figure \ref{fig: CPI_AB} (left) shows the percentage change of the consumer price index for alcoholic beverages in the U.S. (referred to as CPI-AB onward) from its value 12 months prior \footnote{Data source: FRED \url{https://fred.stlouisfed.org/series/CUSR0000SAF116}.}. CPI-AB is recorded monthly. It measures the overall price level of alcoholic beverages across the U.S.. We can see that CPI-AB keeps increasing compared to its value 12 months prior, as the percentage is always greater than 0. Meanwhile, the monthly alcohol involvement crash rate \footnote{Data source: National Highway Traffic Safety Administration \url{https://www.nhtsa.gov/file-downloads?p=nhtsa/downloads/FARS/}.}, according to Figure \ref{fig: CPI_AB} (right), has an overall decreasing trend. We include the CPI-AB sequence as an exogenous variable when fitting models for the monthly alcohol involvement automobile crash rate data set. According to Figure \ref{fig: CPI_AB}, the behavior of CPI-AB before and after August 2019 is different \footnote{A piece of news reporting the change can be found at \url{https://www.businessinsider.com/beer-prices-drop-away-from-home-2019-6}.}. The increase rate of CPI-AB first drops after August 2019, then the rate increases from April 2020. That is, the exogenous variable contains change points.
 \begin{figure}[!h]
 \label{fig: CPI_AB}
 \centering
     \includegraphics[width = 0.45\textwidth]{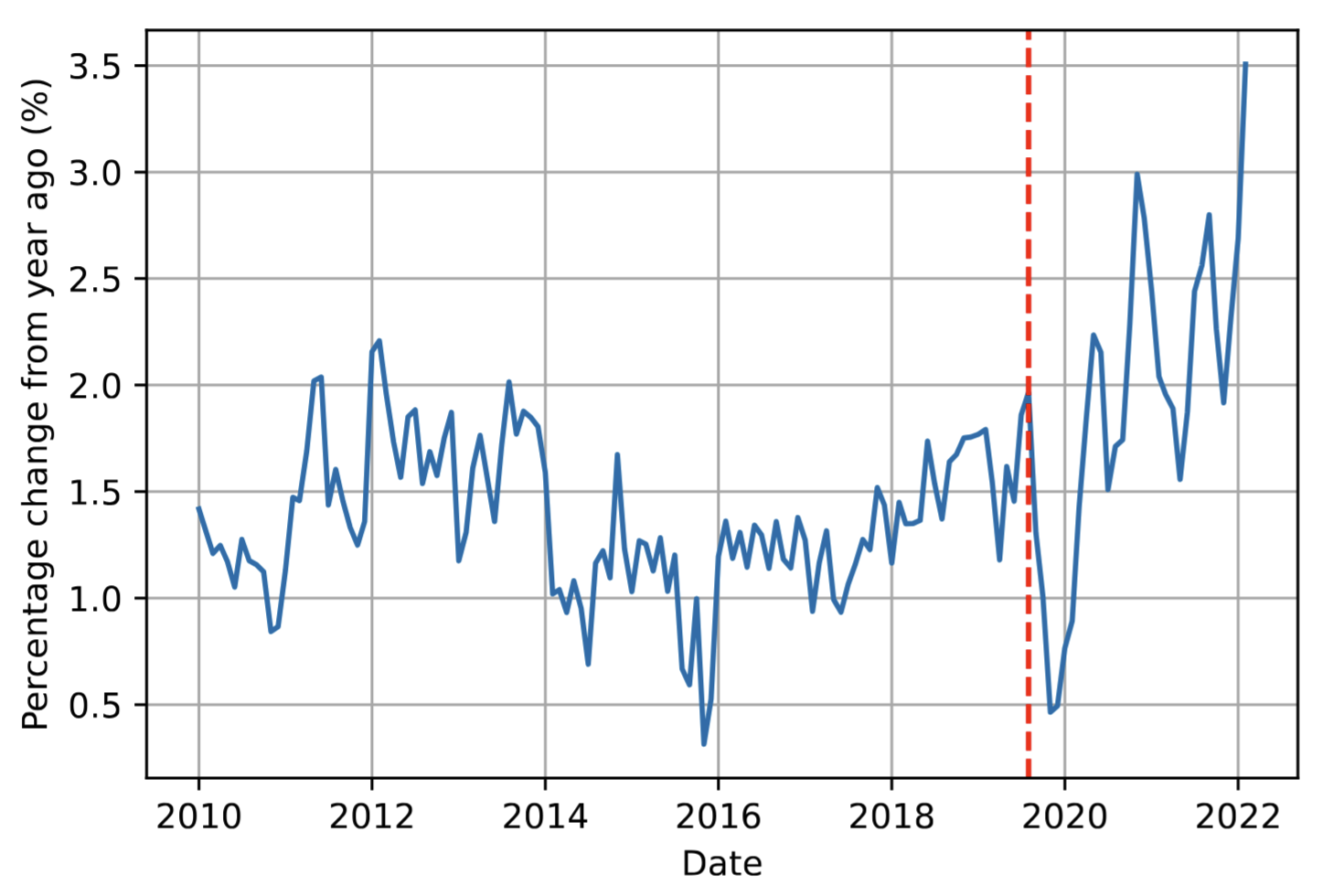}
     \includegraphics[width = 0.45\textwidth]{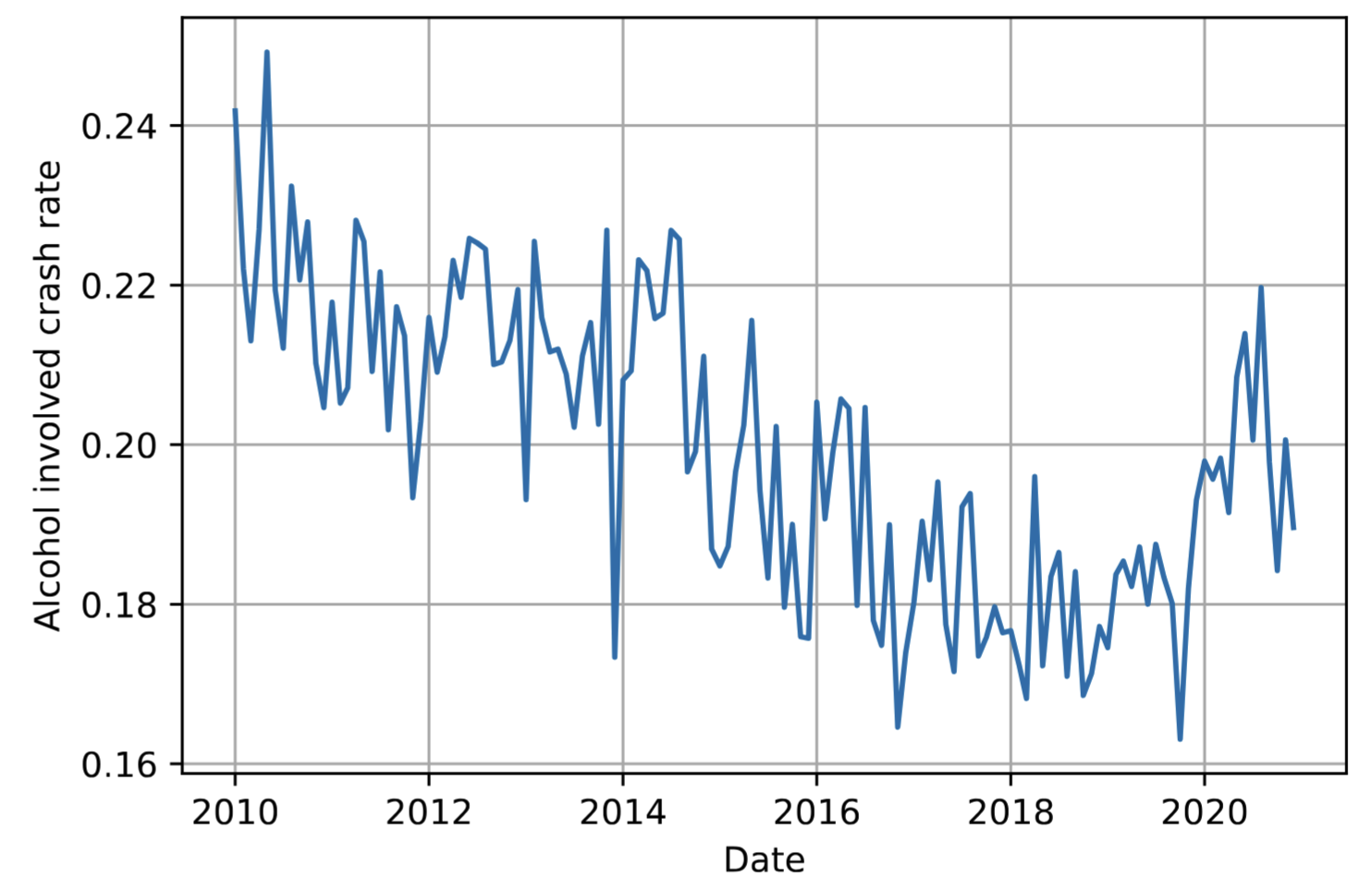}     
     \caption{(Left) percentage change of monthly CPI-AB from January 2010 to February 2022. The red dotted line indicates the location of the CPI-AB in August 2019, where the structure of the time series changes. (Right) monthly alcohol involvement crash rate from January 2010 to December 2020.}
 \end{figure}

We first conduct seasonality and time-trend adjustments to the monthly alcohol involvement automobile crash rate. A regression model with independent variables 'year' and 'month' is fitted on the monthly alcohol involvement automobile crash rate from January 2010 to December 2018. The detrended monthly alcohol involvement automobile crash rate from January 2010 to December 2020 is the original sequence minus the seasonality and time-trend factors from the fitted regression model. The following experiment is applied to the detrended monthly alcohol involvement automobile crash rate. The detrended time series is shown in Figure \ref{fig: adjusted alcohol involvement crash rate}.

We first fit generalized Beta AR($p$) models with $p = 1,2,3,4$ to the detrended monthly alcohol involvement automobile crash rate from January 2010 to December 2018. For each $p$, we calculate the AICs and mean absolute errors (MAEs), $\frac{1}{m}\sum_{t=1}^m |X_t-\hat{\mu}_t|$, of the fitting results of the model without the exogenous variable ($q=0$) and models with order $q$ exogenous variables, $q = 1,2,3,...,8.$ Comparing the model fitting metrics of all the models, we are able to select the most appropriate model for the output sequence. By comparing the metrics of the models with $q=0$, we can verify whether, in this case, the CPI-AB sequence can improve the prediction accuracy of the crash rate sequence. Table \ref{table: real-life data set model comparison} lists the AICs and MAEs of the models with $p=1,2,3,4$ and $q = 0,1,2,4,8$. For every $p$, compared with the model with $q=0$, models with exogenous variables have lower AICs and MAEs. Therefore, we can verify that the CPI-AB sequence can indeed increase the model fitting performance. For every $p$, the fitting result reaches the best when $q=8$ and the AIC starts to increase when $q>8$. Analyzing all the fitting metrics in Table \ref{table: real-life data set model comparison}, we consider that the generalized Beta AR model with $p=3$ and $q=8$ is an appropriate model for the alcohol involvement crash rate as it has the lowest AIC and a relatively small MAE. 

\begin{table}[!h]
\centering
\caption{The AICs and MAEs of the model fitting results.}
\begin{tabular}{r|r|r|r|r|r}
\hline
AIC                    & \multirow{2}{*}{$q=0$} & \multirow{2}{*}{$q=1$} & \multirow{2}{*}{$q=2$} & \multirow{2}{*}{$q=4$} & \multirow{2}{*}{$q=8$} \\ \cline{1-1}
MAE                    &                        &                        &                        &                        &                        \\ \hline
\multirow{2}{*}{$p=1$} & -601.8705              & -602.0930              & -600.2870              & -602.7968              & -607.3994              \\ \cline{2-6} 
                       & 0.01132                & 0.01115                & 0.01116                & 0.01110                & 0.01083                \\ \hline
\multirow{2}{*}{$p=2$} & -619.0036              & -618.8485              & -616.9964              & -620.2313              & -622.6874              \\ \cline{2-6} 
                       & 0.01091                & 0.01070                & 0.01072                & 0.01037                & 0.01017                \\ \hline
\multirow{2}{*}{$p=3$} & -634.8317              & -635.2870              & -633.8645              & -634.2300              & \textcolor{red}{-636.1937}           \\ \cline{2-6} 
                       & 0.00979                & 0.00961                & 0.00958                & 0.00933                & 0.00939                \\ \hline
\multirow{2}{*}{$p=4$} & -601.5165              & -634.4343              & -632.8803              & -633.0470              & -635.1787              \\ \cline{2-6} 
                       & 0.01116                & 0.00962                & 0.00963                & 0.00934                & 0.00932                \\ \hline
\end{tabular}
\label{table: real-life data set model comparison}
\end{table}

Figure \ref{fig: crash_rate_prediction} shows the model fitting results of the model with $p = 3$ and $q = 8$. To show the fitting results clearly, in Figure \ref{fig: crash_rate_prediction}, we show the original sequence and the adjusted conditional mean, which is the conditional mean from the fitted model plus the seasonality and time-trend factors from the fitted regression model.

\begin{figure}
     \centering
     \includegraphics[width = 0.5\textwidth]{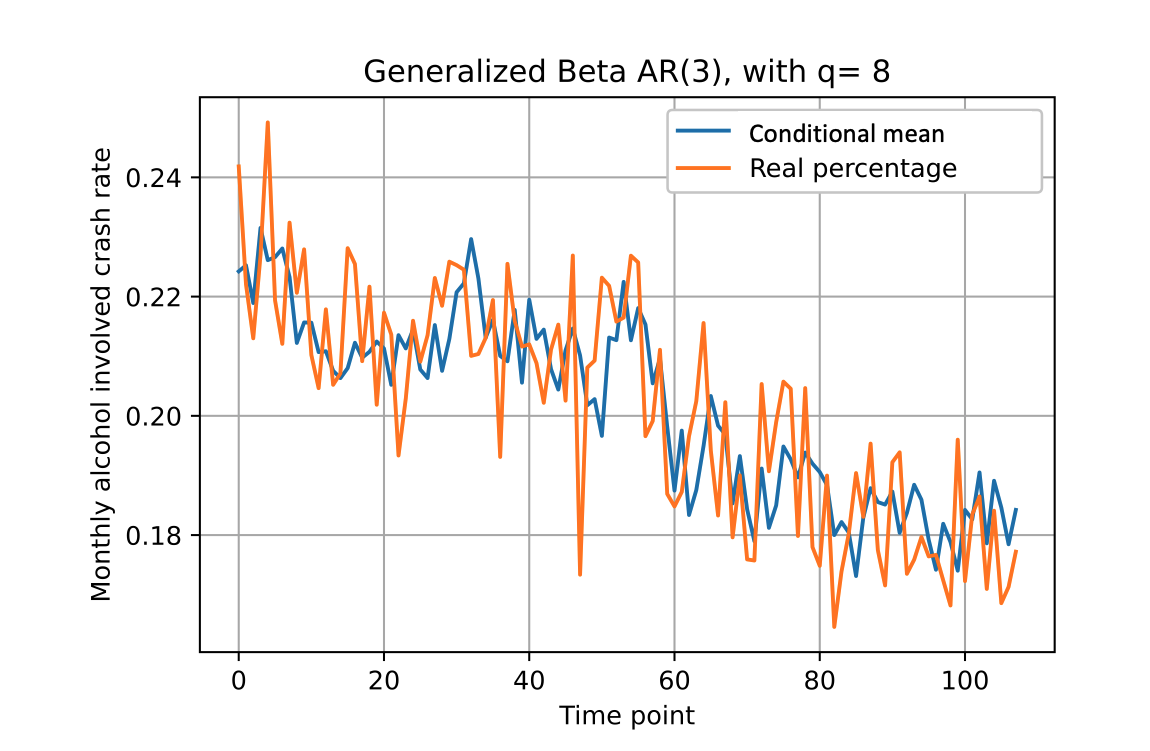}
     \caption{Fitting result for the monthly alcohol involvement automobile crash rate from January 2010 to December 2018 from the model with $p=3$ and $q=8$.}
     \label{fig: crash_rate_prediction}
\end{figure}

We consider the data from January 2010 to December 2018 as the initial observations. So $m = 108$. The monitoring procedure starts from January 2019, and the endpoint is December 2020. That is, the total number of new observations is 24, and $N = 24/108 = 0.22.$ We sample 10 percentages with equal distance, and find the corresponding quantile vector $\bm{x}$ of the training sequence (the monthly alcohol involvement automobile crash rate from January 2010 to December 2018). We choose $\bm{A}$ to be the 10-dimensional diagonal matrix with diagonal elements 1/10. With $N$, $\bm{x}$ and $\bm{A}$ defined above and $\gamma = 0, 0.25, 0.4$, we simulate the Gaussian process-related statistic $\underset{1<s \leq N+1}{\sup} \rho^2(s, \gamma)D_C(s ,\bm{x})^T\bm{A}D_C(s ,\bm{x})$ 10,000 times. For each $\gamma$, we find the threshold with significance level $\alpha = 0.05$, i.e. $c(\gamma, 0.05)$ as shown in Table \ref{table: c(gamma, 0.05)}.

\begin{table}[!h]
\centering
\caption{$c(\gamma, 0.05)$ under different $\gamma$}
\begin{tabular}{l|rrr}
\hline
$\gamma$     & 0          & 0.25       & 0.4        \\ \hline
$c(\gamma, 0.05)$ & 0.1546 & 0.4052 & 0.7812 \\ \hline
\end{tabular}
\label{table: c(gamma, 0.05)}
\end{table}

The detection process will depend on the simulated thresholds. The nonparametric change-point detection is applied to the converted time series. The detection results are listed in Table \ref{table: real-life nonpara results}. We can see that the earliest detected change point is March 2020. The detrended process and the location of March 2020 are shown in Figure \ref{fig: adjusted alcohol involvement crash rate}. 

\begin{table}[!h]
\centering
\caption{Detected change points for the alcohol involvement crash rate.}
\begin{tabular}{l|lll}
\hline
$\gamma$     & 0          & 0.25       & 0.4        \\ \hline
Detected change point & 04/2020 & 03/2020 & 03/2020 \\ \hline
\end{tabular}
\label{table: real-life nonpara results}
\end{table}

\begin{figure}[!h]
    \centering
    \includegraphics[width = 0.5\textwidth]{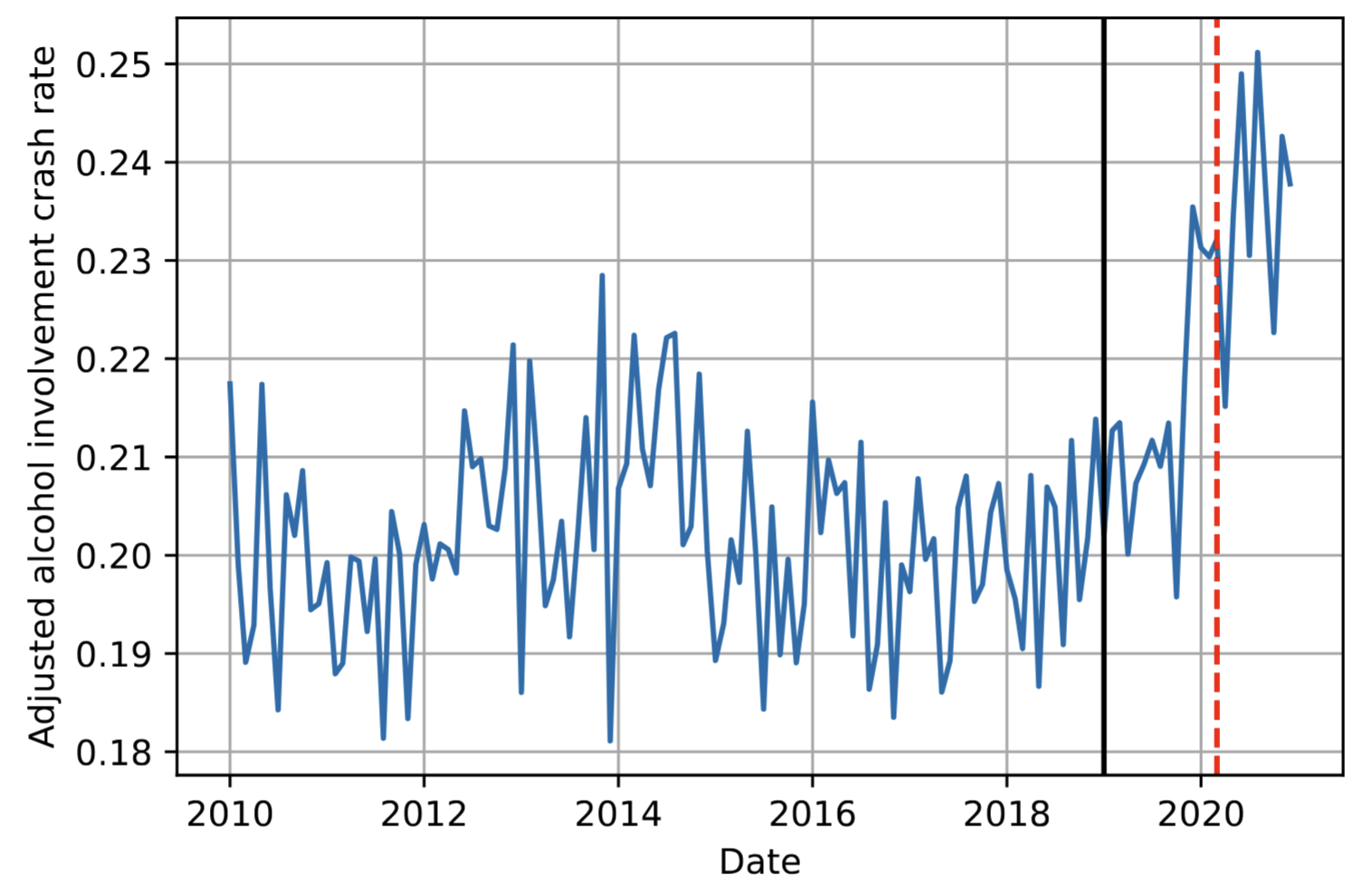}
    \caption{Detrended alcohol involvement crash rate from 2010 to 2020. The black solid line shows the location when the detection is launched (January 2019). The red dotted line shows the location of the detected change point (March 2020).}
    \label{fig: adjusted alcohol involvement crash rate}
\end{figure}

By conducting the nonparametric sequential change-point detection, we can be aware of a potential change happening in March 2020. Therefore, the model can't be used anymore. Meanwhile, researchers can be alerted and investigate the reasons causing the change point.

\newpage

\section*{Appendix}
\subsection*{Appendix 1. Proof of Theorem \ref{theorem:stationary&ergodic}}
\label{sec: proof of Theorem stationary&ergodic}
According to Theorem 15.0.1 of \cite{meyn1993markov}, for a $\psi$-irreducible and aperiodic Markov chain, if there exists a petite set $C$, constant $b <\infty$, $\beta>0$ and a function $V \geq 1$ finite at some $y_0 \in \Omega$ such that the drift condition holds 
\begin{equation*}
    \Delta V(x) \leq -\beta V(x) +b \mathbbm{1}_C(x), x \in \Omega,
\end{equation*}
then the process is geometrically ergodic, where $\Delta V(x)$ is defined as
\begin{equation*}
    \Delta V(x) =  \int P(x,\text{d}y)V(y)-V(x),
\end{equation*}
and $P(x, \text{d}y)$ is the one-step transition pdf from $x$ to $y$.

From Lemma \ref{lemma:weak Feller} we have proven that the state space $\Omega$ is a small (also petite) set. Define $V(\cdot)$ as the constant function $\Omega \rightarrow b$, and let $\beta = 1$, then the drift condition can be easily satisfied.
\begin{equation*}
    \begin{aligned}
    \Delta V(y) = & \int P(y,\text{d}y')V(y')-V(y) \\
             \leq & \int P(y,\text{d}y')(\underset{y' \in \Omega} {\max}V(y'))-V(y) \\
             =& (\underset{y' \in \Omega} {\max}V(y'))-V(y) \\
              = & b-V(y) \\
              = & -V(y)+b\mathbbm{1}_{\Omega}(y), \text{ for any }y \in \Omega.
    \end{aligned}
\end{equation*}
So $\{\bm{Y}_t\}$ is geometrically ergodic. Also, according to Theorem 15.0.1 in \cite{meyn1993markov}, $\psi$-irreducibility, aperiodicity and the drift condition together imply the existence of the unconditional probability of the process, say $\pi$, such that
\begin{equation*}
    \pi(A) = P(\bm{Y}_t \in A), \text{ for any } t\in \mathbbm{Z} \text{ and }A \in \mathcal{B}(\Omega).
\end{equation*}
According to Assumption \ref{assumption: properties of W_t}, $W_t$ is strictly stationary, which means that the transition probabilities of $\{W_t\}$ are free from $t$. The modeling procedures ((\ref{equation: model definition}) and (\ref{eq: link function of GBETAAR2})), together with the strict stationarity of $W_t$, imply that the transition probabilities of $\{\bm{Y}_t\}$ are free from $t$. 
Therefore, we can conclude that $\{\bm{Y}_t\}$ is strictly stationary.

\subsection*{Appendix 2. Proof of Proposition \ref{prop: X_t is strong mixing}}
\label{sec: proof of X_t is strong mixing}
\begin{proof}
According to Theorem 3.7 of \cite{bradley2005basic}, a strictly stationary and geometrically ergodic process satisfies absolute regularity with absolute regularity coefficient $\beta(n)\rightarrow 0$ at least exponentially fast as $n \rightarrow \infty$. Shown by (1.11) in \cite{bradley2005basic}, the strong mixing coefficient is bounded by the absolute regularity coefficient. Absolute regularity implies strong mixing. So $\{\bm{Y}_{zl}\}$, shown as a strictly stationary and geometrically ergodic process in Theorem \ref{theorem:stationary&ergodic}, is strong mixing with at least an exponential mixing rate. Define the strong mixing coefficient as
\begin{equation*}
    \alpha(n) = \underset{j \in \mathbbm{N}^+}{\sup}\;\alpha(\mathcal{Y}_0^j,\mathcal{Y}_{j+n}^\infty) = \underset{j \in \mathbbm{N}^+}{\sup} \; |P(A \cap B)-P(A)P(B)|, \text{ for } \forall A \in \mathcal{Y}_0^j, B \in \mathcal{Y}_{j+n}^\infty,
\end{equation*}
where $\mathcal{Y}_J^L := \sigma(\bm{Y}_{zl}, J \leq z \leq L).$
Given the $\{Y_t\}$ is a Markovian realization of $\{X_t\}$, $\{X_t\}$ is strong mixing with at least an exponential mixing rate.
\end{proof}

\subsection*{Appendix 3. Proof of Theorem \ref{theorem: asymptotic dist of D_C under the null}}
\label{sec: proof of asymptotic dist of D_C under the null}
\begin{proof}
We have illustrated that $\{X_t\}$ satisfies the precondition in Proposition \ref{proposition: asymptotic distribution of CDF}. Therefore, $\mathbbm{B}_m(s,x) \overset{\mathcal{D}}{\rightarrow} \mathbbm{B}_C(s,x)$ for any $0<s<N+1$ and $x \in \Omega.$ By Continuous Mapping Theorem, $\mathbbm{B}_m(s,x)-s\mathbbm{B}_m(1,x) \overset{\mathcal{D}}{\rightarrow} \mathbbm{B}_C(s,x)-s\mathbbm{B}_C(1,x).$ (\ref{eq: asymptotic distribution of the test statistic}) can then be derived as the vector version of the asymptotic behavior. As linear transformations of the Gaussian process $\mathbbm{B}_C$, the covariance of $D_C(s,x_i)$ and $D_C(s,x_j)$
\begin{equation*}
    \begin{aligned}
         \text{Cov}(D_C(s,x_i), D_C(s,x_j)) &= \text{Cov}(\mathbbm{B}_C(s, x_i)-s\mathbbm{B}_C(1,x_i),\mathbbm{B}_C(s, x_j)-s\mathbbm{B}_C(1,x_j)) \\
         &= \text{Cov}(\mathbbm{B}_C(s, x_i), \mathbbm{B}_C(s, x_j))-s\text{Cov}(\mathbbm{B}_C(1, x_i), \mathbbm{B}_C(s, x_j))\\
         & \quad -s \text{Cov}(\mathbbm{B}_C(s, x_i), \mathbbm{B}_C(1, x_j)))+s^2 \text{Cov}(\mathbbm{B}_C(1, x_i), \mathbbm{B}_C(1, x_j))) \\
         & = s \Gamma(x_i,x_j) -s \min\{s,1\}\Gamma(x_i,x_j)\\
         & \quad -s \min\{s,1\}\Gamma(x_i,x_j)+s^2 \Gamma(x_i,x_j) \\
         &= s(s-1) \Gamma(x_i,x_j), \text{ for } 1<s\leq N+1.
    \end{aligned}
\end{equation*}
So the covariance matrix of $D_C(s,\bm{x})$ is $s(s-1)\Gamma(\bm{x}).$
\end{proof}

\subsection*{Appendix 4. Proof of Theorem \ref{theorem: power 1}}
\label{sec: proof of power 1}
\begin{proof}

For every $x_i \in \bm{x}$, as defined in (\ref{eq: definition of the test statistic}),
\begin{equation}
\label{eq: decompose D_m(s,u)}
    \begin{aligned}
         D_m(s_0, x_i) &= \sqrt{m}s_0(\hat{F}_{1:\lfloor ms_0 \rfloor}(x_i)-\hat{F}_{1:m}(x_i)) \\ 
         &= \frac{1}{\sqrt{m}}\sum_{t=1}^{\lfloor ms_0 \rfloor} \mathbbm{1}(X_t \leq x_i)-\sqrt{m}s_0 \hat{F}_{1:m}(x_i) \\
         &= \frac{1}{\sqrt{m}}(\sum_{t=1}^{\lfloor ms^* \rfloor} \mathbbm{1}(X_t \leq x_i)+\sum_{t=\lfloor ms^* \rfloor+1}^{\lfloor ms_0 \rfloor} \mathbbm{1}(X_t \leq x_i))-\sqrt{m}(s^*+(s_0-s^*)) \hat{F}_{1:m}(x_i) \\
         &= \sqrt{m}s^* \hat{F}_{1:\lfloor ms^* \rfloor}(x_i)+\sqrt{m}(s_0-s^*) \hat{F}_{\lfloor ms^* \rfloor+1:\lfloor ms_0 \rfloor}(x_i)-\sqrt{m}(s^*+(s_0-s^*)) \hat{F}_{1:m}(x_i) \\ 
         &=\sqrt{m}s^*(\hat{F}_{1:\lfloor ms^* \rfloor}(x_i)-\hat{F}_{1:m}(x_i))+\sqrt{m}(s_0-s^*)(\hat{F}_{\lfloor ms^* \rfloor+1:\lfloor ms_0 \rfloor}(x_i)-\hat{F}_{1:m}(x_i)) \\
         &= \sqrt{m}\Big(s^*(\hat{F}_{1:\lfloor ms^* \rfloor}(x_i)-\hat{F}_{1:m}(x_i))\\
         & \quad \quad \quad +(s_0-s^*)(\hat{F}_{\lfloor ms^* \rfloor+1:\lfloor ms_0 \rfloor}(x_i)-{F}'(x_i))\\
         & \quad \quad \quad +(s_0-s^*)({F}'(x_i)-\hat{F}_{1:m}(x_i)) \Big),
    \end{aligned}
\end{equation}
where $F'(x_i)$ is the true CDF of the process  at $x_i$ after the change point. $\{X_t\}$ is proven to be stationary in Theorem \ref{theorem:stationary&ergodic}.
So according to Proposition \ref{prop: General Glivenko–Cantelli}, as $m\rightarrow \infty$, for any $x_i \in \Omega$, 
\begin{equation*}
    |\hat{F}_{1:\lfloor ms^* \rfloor}(x_i)-F_0(x_i)| = o_p(1),
\end{equation*}
and 
\begin{equation*}
    |\hat{F}_{1: m}(x_i)-F_0(x_i)| = o_p(1).
\end{equation*}
Therefore, 
\begin{equation*}
\begin{aligned}
     &s^*|\hat{F}_{1:\lfloor ms^* \rfloor}(x_i)-\hat{F}_{1:m}(x_i)| \\ \leq & s^*(|\hat{F}_{1:\lfloor ms^* \rfloor}(x_i)-F_0(x_i)|+|\hat{F}_{1: m}(x_i)-F_0(x_i)|) \\
     =& O(1) \times (o_p(1)+o_p(1)) \\
     =& o_p(1).
\end{aligned}
\end{equation*} 
Given that the sequence after the change point is still stationary (Assumption \ref{assumption: power analysis}(c)), according to Proposition \ref{prop: General Glivenko–Cantelli}, we can also get $(s_0-s^*)|\hat{F}_{\lfloor ms^* \rfloor+1:\lfloor ms_0 \rfloor}(x_i)-F'(x_i)| = o_p(1)$ as $m \rightarrow \infty$. Meanwhile, 
\begin{equation*}
    \begin{aligned}
         &(s_0-s^*)(F'(x_i)-\hat{F}_{1:m}(x_i)) \\
         =&(s_0-s^*)(F'(x_i)-F_0(x_i))+(s_0-s^*)(F_0(x_i)-\hat{F}_{1:m}(x_i)) \\
         =& (s_0-s^*)(F'(x_i)-F_0(x_i))+o_p(1).
    \end{aligned}
\end{equation*}
From the decomposition of $D_m(s,x_i)$ in (\ref{eq: decompose D_m(s,u)}), we can see that
\begin{equation*}
\begin{aligned}
         & \sqrt{m}\Big(s^*(\hat{F}_{1:\lfloor ms^* \rfloor}(x_i)-\hat{F}_{1:m}(x_i))+(s_0-s^*)(\hat{F}_{\lfloor ms^* \rfloor+1:\lfloor ms_0 \rfloor}(x_i)-F'(x_i))+(s_0-s^*)(F'(x_i)-\hat{F}_{1:m}(x_i)) \Big)\\
    =& \sqrt{m}\Big(o_p(1)+o_p(1) + (s_0-s^*)(F'(x_i)-F_0(x_i))+o_p(1)\Big) \\
    =&\sqrt{m}\Big((s_0-s^*)(F'(x_i)-F_0(x_i))+o_p(1)\Big),
\end{aligned}
\end{equation*}
The multivariate test statistic can then be rewritten as  
\begin{equation*}
    D_m(s_0,\bm{x}) = \sqrt{m}\Big((s_0-s^*)(F'(\bm{x})-F_0(\bm{x}))+o_p(1)\Big).
\end{equation*}

Given Assumption \ref{assumption: power analysis}(d), $F'(\bm{x})-F_0(\bm{x}) \neq \bm{0}.$ Combining it with the positive definiteness of $\bm{A}$, 
\begin{equation*}
    (F'(\bm{x})-F_0(\bm{x}))^T \bm{A} (F'(\bm{x})-F_0(\bm{x}))>0.
\end{equation*}
Then, as $m \rightarrow \infty$, the test statistic 
\begin{equation*}
\begin{aligned}
&\underset{1<s \leq N+1}{\sup} \rho^2(s, \gamma)D_m(s ,\bm{x})^T\bm{A}D_m(s ,\bm{x}) \\
 \geq& \;\rho^2(s_0, \gamma)D_m(s_0,\bm{x})^T\bm{A}D_m(s_0,\bm{x}) \\
 = & \rho^2(s_0, \gamma)m (s_0-s^*)^2 (F'(\bm{x})-F_0(\bm{x})+o_p(1))^T \bm{A} (F'(\bm{x})-F_0(\bm{x})+o_p(1)) \\
 \rightarrow & \infty.
\end{aligned}
\end{equation*}
So the equation holds that
\begin{equation*}
    \underset{m\rightarrow \infty}{\lim}P(\underset{1<s \leq N+1}{\sup} \rho^2(s, \gamma)D_m(s ,\bm{x})^T\bm{A}D_m(s ,\bm{x}) \geq c |H_a) = 1.
\end{equation*}
\end{proof}

\subsection*{Appendix 5. Simulation Study}
\subsubsection*{Asymptotic Results under $H_0$}
\label{sec: Asymptotic results under $H_0$}
To verify whether the asymptotic behavior of the test statistic under the $H_0$ is as proposed in Theorem \ref{theorem: asymptotic distribution under the null hypothesis}, we run simulations to compare the empirical rejection probabilities and the nominal significance levels. 

In this experiment, the exogenous variable sequence follows 
\begin{equation}
\label{eq: simulation model under the null1}
   W_t  = -0.1 W_{t-1}+\epsilon_{t}, \quad \epsilon_t \overset{i.i.d.}{\sim}N(0,1),
\end{equation}
and there is a Beta AR(3) sequence with $\tau = 100$ and the conditional mean $\mu_t$ determined by 
\begin{equation}
\label{eq: simulation model under the null2}
    \log(\mu_t/(1-\mu_t)) = \varphi_0+ \sum_{i=1}^3 \varphi_i \log(X^*_{t-i}/(1-X^*_{t-i}))+ 0.5 W_t^*, 
\end{equation}
with $[\varphi_0, \varphi_1, \varphi_2, \varphi_3] = [0.5, 0.1,0.2,0.2]$. We pick 20 percentages with equal distance and find the corresponding quantiles of the model (\ref{eq: simulation model under the null1}) and (\ref{eq: simulation model under the null2}). Since the distribution of $X_t$ is unknown and is not even, compared to choosing $\bm{x}$, starting from choosing percentages and then finding the corresponding $\bm{x}$ enables us to adjust the precision of $\bm{x}$. In this simulation, the sampled percentage vector $\bm{u} = [0.048, 0.095,...,0.952]^T \in \mathbbm{R}^{20}$. Then we first simulate a sequence of 10,000 data points following the model (\ref{eq: simulation model under the null1}) and (\ref{eq: simulation model under the null2}), and find the sampled value vector $\bm{x}$ corresponding to $\bm{u}$ from the sequence. For example, $x_1$ in $\bm{x}$ is the $(10000 \times 0.048)^{\text{th}}$ order statistic of the sequence.
Based on $\bm{x}$, we can estimate $\Gamma(\bm{x})$ using the 10,000 data points following the method we discuss in Proposition \ref{proposition: asymptotic distribution of CDF}. We set up the cut point $t^* = 50.$

We choose $N = 2$ for this experiment. With the three tuning parameters $\gamma = 0, 0.25, 0.4$, we generate Gaussian processes  and build $D_C(s, \bm{x}), 1 < s \leq N+1$ as defined in (\ref{eq: asymptotic distribution of the test statistic}), and calculate the thresholds $c(\gamma,\alpha)$ so that 
\begin{equation*}
    P(\underset{1<s \leq N+1}{\sup} \rho^2(s, \gamma)D_C(s ,\bm{x})^T\bm{A}D_C(s ,\bm{x}) \leq c(\gamma,\alpha)) = 1-\alpha,
\end{equation*}
for the four significance levels $\alpha = 0.1, 0.05, 0.025, 0.01$. We use the weight function $\rho(s,\gamma)$ in (\ref{eq: definition of the weight function}) and choose $\delta = 0.0001\footnote{A random small enough constant.}$. We choose $\bm{A}$ as the 20-dimensional diagonal matrix with elements equal to 1/20 in order to control the scale of $D_C(s, \bm{x})^T\bm{A}D_C(s, \bm{x})$. 

In Proposition \ref{proposition: asymptotic distribution of CDF}, we discuss the method we use to generate Gaussian processes in practice. We choose $m = 1000$ and generate $(N+1)*m = 3000$ independent normal random variables following $N(\bm{0}, \Gamma(\bm{x}))$. The cumulative sum of the first $k, k = 1,2,...,(N+1)m,$ random variables divided by $\sqrt{m}$ is a realization of $\mathbbm{B}_C(k/m, \bm{x})$, the first element in $D_C(s,\bm{x}) (s=k/m)$. $\mathbbm{B}_C(1, \bm{x})$, the stochastic process part of the second element in $D_C(s,\bm{x})$ is realized when $k = m$. Since we require $s \in (1,N+1]$, for every iteration, we repeat the procedure to build $\rho^2(s, \gamma)D_C(s ,\bm{x})^T\bm{A}D_C(s ,\bm{x})$, and record the supremum over $s$ for $s = \frac{m+1}{m}, \frac{m+1}{m},...,\frac{(N+1)m}{m}.$ We repeat the process 10,000 times, and calculate 
the $(1-\alpha)$ quantile of the supremum of the 10,000 samples as the threshold $c(\gamma, \alpha)$.
In Table \ref{table:null c }, we list the threshold $c(\gamma,\alpha)$ for different $(\gamma,\alpha)$.

\begin{table}[!htb]
\centering
\caption{$c(\gamma,\alpha)$ under different $\gamma$ and significance level $\alpha$}
\begin{tabular}{lllll}
\hline
$\alpha$        & 0.1    & 0.05   & 0.025  & 0.01    \\ \hline
$\gamma = 0$    & 0.7251 & 0.9507 & 1.1587 & 1.4548  \\
$\gamma = 0.25$ & 1.0178 & 1.2828 & 1.5665 & 1.9026  \\
$\gamma = 0.4$  & 1.3774 & 1.7059 & 2.0232 & 2.4226 \\ \hline
\end{tabular}
\label{table:null c }
\end{table}
We set up three experiments with the numbers of initial observations $m = 50, 100, 150$ and the same $N = 2$. For each $m$, we simulate 5,000 samples and calculate the empirical rejection probability, which is the percentage of cases that 
\begin{equation*}
    \underset{1<s \leq N+1}{\sup} \rho^2(s, \gamma)D_m(s ,\bm{x})^T\bm{A}D_m(s ,\bm{x}) \geq c(\gamma, \alpha). 
\end{equation*}
The results are shown in Table \ref{table:null alpha compare}. 
\begin{table}[!htb]
\centering
\caption{The empirical probabilities with different $m$ and $\gamma$}
\begin{tabular}{llllll}
\hline
\multicolumn{2}{c}{$\alpha$}                                & 0.1    & 0.05                       & 0.025  & 0.01   \\ \hline
\multicolumn{1}{c}{\multirow{3}{*}{$\gamma = 0$}} & $m=50$ & 0.1086 & 0.0518& 0.0262& 0.0120 \\
\multicolumn{1}{c}{}                              & $m=100$ & 0.1210 & 0.0648& 0.0348 & 0.0168 \\
\multicolumn{1}{c}{}                              & $m=150$ & 0.1160 & 0.0554 & 0.0310 & 0.0152 \\ \hline
\multirow{3}{*}{$\gamma = 0.25$}                  & $m=50$ & 0.0960 & 0.0508& 0.0234 & 0.0116 \\
                                                  & $m=100$ & 0.1174 & 0.0644& 0.0328& 0.0160 \\
                                                  & $m=150$ & 0.1068& 0.0596 & 0.0296& 0.0150 \\ \hline
\multirow{3}{*}{$\gamma = 0.4$}                   & $m=50$ & 0.0768 & 0.0366 & 0.0192 & 0.0090 \\
                                                  & $m=100$ & 0.1008 & 0.0528& 0.0280& 0.0124 \\
                                                  & $m=150$ & 0.0954& 0.0538& 0.0278& 0.0134\\ \hline
\end{tabular}
\label{table:null alpha compare}
\end{table}
We can see that the empirical probabilities and the corresponding nominal significance levels are reasonably close, which verifies that the asymptotic distribution of the test statistics under $H_0$ is as proposed in Theorem \ref{theorem: asymptotic distribution under the null hypothesis}.

\subsubsection*{Nonparametric Sequential Detection Results for Distributional Changes}
We simulate an exogenous variable sequence including a distributional change. The first 250 exogenous variables follow
\begin{equation*}
   W_t = -0.1 W_{t-1}+\epsilon_{t}, \quad \epsilon_t \overset{i.i.d.}{\sim}N(0,1),
\end{equation*}

and then the following 250 exogenous variables switch to the model 
\begin{equation*}
        W_t = -\frac{1}{5}W_{t-1}-\frac{3}{10}W_{t-2}+\frac{1}{10}\epsilon_{t}+\frac{1}{20} \epsilon_{t-1}, \quad \epsilon_t \overset{i.i.d.}{\sim}N(0,1).
\end{equation*}
The conditional mean of the total 500 outputs is determined by 
\begin{equation}
\label{eq: link function when parametric fails}
 \log(\mu_t/(1-\mu_t)) =  \varphi_0+ \sum_{i=1}^3 \varphi_i \log(X^*_{t-i}/(1-X^*_{t-i}))+\sum_{j=0}^2 W^*_{t-j} \bm{\phi}_j,
\end{equation}
with $[\varphi_0, \varphi_1, \varphi_2, \varphi_3] = [0.5, 0.1,0.2,0.2]$ and $\phi_i = 0.5$ for $i = 0,1,2$. The dispersion parameter $\tau = 100$. The transformed variables are defined as 
\begin{equation}
\label{eq: transformed W_t}
    W_t^* = \min(\max(-10, W_t), 10),
\end{equation}
and 

\begin{equation}
\label{eq: transformed X_t}
    X_{t}^* = \min(\max(0.001, X_t), 0.999).
\end{equation}
That is, the change point exists in $\{W_t\}$.

We take the first 200 samples as the initial observations, so $N = 1.5$. We sample 5 quantiles for this case. That is, $\bm{u} = [0.167, 0.333, 0.5, 0.667, 0.833]^T \in \mathbbm{R}^5$. The corresponding $\bm{x}$ and the covariance matrix $\Gamma(\bm{x})$ can be found accordingly. We then assign the inverse of $\Gamma(\bm{x})$ to $\bm{A}$ and estimate the 95\% threshold $c(\gamma, 0.05)$ that
\begin{equation*}
    P(\underset{1<s \leq N+1}{\sup} \rho^2(s, \gamma)D_C(s ,\bm{x})^T\bm{A}D_C(s ,\bm{x}) \leq c(\gamma,0.05)) = 0.95,
\end{equation*}
for $\gamma  = 0, 0.25, 0.4.$ Meanwhile, we calculate the average distances between the true change point $k^* = 50$ and the detected change points. The results are summarized as follows.

\begin{table}[!h]
\centering
\caption{The nonparametric detection results.}
\begin{tabular}{l|lll}
\hline
$\gamma$                      & 0   & 0.25 & 0.4  \\ \hline
Empirical rejection prob (\%) & 100 & 100  & 100 \\ 
Average distance & 47.15 & 39.03 & 35.40 \\ \hline
\end{tabular}
\label{table: nonparametric works}
\end{table}
As shown in Table \ref{table: nonparametric works}, all change points are detected by the nonparametric method. It verifies the power analysis result in Section \ref{sec: power of the test statistic} that the power converges to 1. As $\gamma$ increases, the average distance between the true change point and the detected change point decreases, i.e. a large $\gamma$ can increase the sensitivity of the detection method.

\newpage
\bibliography{arxiv} 

\begin{thebibliography}{10}

\bibitem{park1987nonparametric}
C.~Park, C.~Park, M.~R. Reynolds~Jr, and M.~R. Reynolds~Jr, ``Nonparametric procedures for monitoring a location parameter based on linear placement statistics,'' {\em Sequential Analysis}, vol.~6, no.~4, pp.~303--323, 1987.

\bibitem{bakir2004distribution}
S.~T. Bakir, ``A distribution-free shewhart quality control chart based on signed-ranks,'' {\em Quality Engineering}, vol.~16, no.~4, pp.~613--623, 2004.

\bibitem{bakir2006distribution}
S.~T. Bakir, ``Distribution-free quality control charts based on signed-rank-like statistics,'' {\em Communications in Statistics-Theory and Methods}, vol.~35, no.~4, pp.~743--757, 2006.

\bibitem{alloway1991control}
J.~A. Alloway~Jr and M.~Raghavachari, ``Control chart based on the hodges-lehmann estimator,'' {\em Journal of quality Technology}, vol.~23, no.~4, pp.~336--347, 1991.

\bibitem{amin1995nonparametric}
R.~W. Amin, M.~R. Reynolds~Jr, and B.~Saad, ``Nonparametric quality control charts based on the sign statistic,'' {\em Communications in Statistics-Theory and Methods}, vol.~24, no.~6, pp.~1597--1623, 1995.

\bibitem{vasilopoulos1978modification}
A.~V. Vasilopoulos and A.~Stamboulis, ``Modification of control chart limits in the presence of data correlation,'' {\em Journal of Quality Technology}, vol.~10, no.~1, pp.~20--30, 1978.

\bibitem{hussein2013comparison}
M.~Hussein, A.~Allawi Al-Morshedi, and H.~Shomran, ``The comparison between shewhart control chart, cusum and ewma,'' 2013.

\bibitem{hack11992new}
P.~Hack1 and J.~Ledolter, ``A new nonparametric quality control technique,'' {\em Communications in Statistics-Simulation and Computation}, vol.~21, no.~2, pp.~423--443, 1992.

\bibitem{hackl1991control}
P.~Hackl and J.~Ledolter, ``A control chart based on ranks,'' {\em Journal of Quality Technology}, vol.~23, no.~2, pp.~117--124, 1991.

\bibitem{amin1991nonparametric}
R.~W. Amin and A.~J. Searcy, ``A nonparametric exponentially weighted moving average control scheme,'' {\em Communications in Statistics-Simulation and Computation}, vol.~20, no.~4, pp.~1049--1072, 1991.

\bibitem{liu2013sequential}
L.~Liu, X.~Zi, J.~Zhang, and Z.~Wang, ``A sequential rank-based nonparametric adaptive ewma control chart,'' {\em Communications in Statistics-Simulation and Computation}, vol.~42, no.~4, pp.~841--859, 2013.

\bibitem{zou2012spatial}
C.~Zou, Z.~Wang, and F.~Tsung, ``A spatial rank-based multivariate ewma control chart,'' {\em Naval Research Logistics (NRL)}, vol.~59, no.~2, pp.~91--110, 2012.

\bibitem{schmid1997ewma}
W.~Schmid, ``On ewma charts for time series,'' in {\em Frontiers in statistical quality control}, pp.~115--137, Springer, 1997.

\bibitem{koehler2001ewma}
A.~B. Koehler, N.~B. Marks, and R.~T. O'connell, ``Ewma control charts for autoregressive processes,'' {\em Journal of the Operational Research Society}, vol.~52, no.~6, pp.~699--707, 2001.

\bibitem{schmid1997some}
W.~Schmid and A.~Schone, ``Some properties of the ewma control chart in the presence of autocorrelation,'' {\em The Annals of Statistics}, pp.~1277--1283, 1997.

\bibitem{pawlak2004detecting}
M.~Pawlak, E.~Rafaj{\l}owicz, and A.~Steland, ``On detecting jumps in time series: nonparametric setting,'' {\em Journal of Nonparametric Statistics}, vol.~16, no.~3-4, pp.~329--347, 2004.

\bibitem{page1954continuous}
E.~S. Page, ``Continuous inspection schemes,'' {\em Biometrika}, vol.~41, no.~1/2, pp.~100--115, 1954.

\bibitem{page1955test}
E.~Page, ``A test for a change in a parameter occurring at an unknown point,'' {\em Biometrika}, vol.~42, no.~3/4, pp.~523--527, 1955.

\bibitem{kim2006new}
S.-H. Kim, C.~Alexopoulos, D.~Goldsman, and K.-L. Tsui, ``A new model-free cusum procedure for autocorrelated processes,'' {\em Technical Report}, 2006.

\bibitem{azizzadeh2014cusum}
F.~Azizzadeh and S.~Rezakhah, ``The cusum test for detecting structural changes in strong mixing processes,'' {\em Communications in Statistics-Theory and Methods}, vol.~43, no.~17, pp.~3733--3750, 2014.

\bibitem{vera2004comparative}
C.~Vera~do Carmo, L.~F.~D. Lopes, and A.~M. Souza, ``Comparative study of the performance of the cusum and ewma control charts,'' {\em Computers \& Industrial Engineering}, vol.~46, no.~4, pp.~707--724, 2004.

\bibitem{hawkins2014cusum}
D.~M. Hawkins and Q.~Wu, ``The cusum and the ewma head-to-head,'' {\em Quality Engineering}, vol.~26, no.~2, pp.~215--222, 2014.

\bibitem{zwetsloot2017head}
I.~M. Zwetsloot and W.~H. Woodall, ``A head-to-head comparative study of the conditional performance of control charts based on estimated parameters,'' {\em Quality Engineering}, vol.~29, no.~2, pp.~244--253, 2017.

\bibitem{knoth2004control}
S.~Knoth and W.~Schmid, ``Control charts for time series: A review,'' {\em Frontiers in statistical quality control 7}, pp.~210--236, 2004.

\bibitem{chakraborti2001nonparametric}
S.~Chakraborti, P.~Van~der Laan, and S.~Bakir, ``Nonparametric control charts: an overview and some results,'' {\em Journal of quality technology}, vol.~33, no.~3, pp.~304--315, 2001.

\bibitem{chao2019systematic}
F.~Chao, P.~Gerland, A.~R. Cook, and L.~Alkema, ``Systematic assessment of the sex ratio at birth for all countries and estimation of national imbalances and regional reference levels,'' {\em Proceedings of the National Academy of Sciences}, vol.~116, no.~19, pp.~9303--9311, 2019.

\bibitem{aitchison1982statistical}
J.~Aitchison, ``The statistical analysis of compositional data,'' {\em Journal of the Royal Statistical Society: Series B (Methodological)}, vol.~44, no.~2, pp.~139--160, 1982.

\bibitem{brunsdon1998time}
T.~M. Brunsdon and T.~Smith, ``The time series analysis of compositional data,'' {\em Journal of Official Statistics}, vol.~14, no.~3, p.~237, 1998.

\bibitem{silva2001modelling}
D.~Silva and T.~Smith, ``Modelling compositional time series from repeated surveys,'' {\em Survey methodology}, vol.~27, no.~2, pp.~205--215, 2001.

\bibitem{larrosa2017compositional}
J.~Larrosa, ``Compositional time series: past and perspectives,'' {\em Atlantic Review of Economics}, vol.~1, 2017.

\bibitem{liu2024sequential}
Y.~Liu and B.~Andrews, ``Sequential change-point detection for compositional time series with exogenous variables,'' {\em arXiv preprint arXiv:2402.18130}, 2024.

\bibitem{woodard2010stationarity}
D.~Woodard, D.~Matteson, and S.~Henderson, ``Stationarity of count-valued and nonlinear time series models,'' 2010.

\bibitem{meyn1993markov}
S.~P. Meyn and R.~L. Tweedie, {\em Markov chains and stochastic stability}.
\newblock London: Sringer-Verlag., 1993.

\bibitem{berkovitz2006ergodic}
J.~Berkovitz, R.~Frigg, and F.~Kronz, ``The ergodic hierarchy, randomness and hamiltonian chaos,'' {\em Studies in History and Philosophy of Science Part B: Studies in History and Philosophy of Modern Physics}, vol.~37, no.~4, pp.~661--691, 2006.

\bibitem{kojadinovic2021nonparametric}
I.~Kojadinovic and G.~Verdier, ``Nonparametric sequential change-point detection for multivariate time series based on empirical distribution functions,'' {\em Electronic Journal of Statistics}, vol.~15, no.~1, pp.~773--829, 2021.

\bibitem{bucher2015note}
A.~B{\"u}cher, ``A note on weak convergence of the sequential multivariate empirical process under strong mixing,'' {\em Journal of Theoretical Probability}, vol.~28, no.~3, pp.~1028--1037, 2015.

\bibitem{dette2018likelihood}
H.~Dette and J.~G{\"o}smann, {\em A likelihood ratio approach to sequential change point detection}.
\newblock Universit{\"a}tsbibliothek Dortmund, 2018.

\bibitem{wied2013monitoring}
D.~Wied and P.~Galeano, ``Monitoring correlation change in a sequence of random variables,'' {\em Journal of Statistical Planning and Inference}, vol.~143, no.~1, pp.~186--196, 2013.

\bibitem{athreya2016general}
K.~B. Athreya and V.~Roy, ``General glivenko--cantelli theorems,'' {\em Stat}, vol.~5, no.~1, pp.~306--311, 2016.

\bibitem{chaloupka1993alcohol}
F.~J. Chaloupka, H.~Saffer, and M.~Grossman, ``Alcohol-control policies and motor-vehicle fatalities,'' {\em The Journal of Legal Studies}, vol.~22, no.~1, pp.~161--186, 1993.

\bibitem{ruhm1996alcohol}
C.~J. Ruhm, ``Alcohol policies and highway vehicle fatalities,'' {\em Journal of health economics}, vol.~15, no.~4, pp.~435--454, 1996.

\bibitem{bradley2005basic}
R.~C. Bradley {\em et~al.}, ``Basic properties of strong mixing conditions. a survey and some open questions,'' {\em Probability surveys}, vol.~2, pp.~107--144, 2005.

\end{thebibliography}
\bibliographystyle{ieeetr}

\end{document}